\documentclass[11pt]{article}
    
    \usepackage[margin=1.0in]{geometry}
    \usepackage{mathtools}
    \usepackage{amsthm}
    \usepackage{fancyhdr}
    \usepackage{hyperref} 
    \usepackage{amssymb}
    \usepackage[capitalise]{cleveref}
    \usepackage{subcaption}
    \usepackage{cmbright}
    \usepackage{anyfontsize}
    
    \newtheorem{theorem}{Theorem}[section]
    \newtheorem{corollary}{Corollary}[theorem]
    \newtheorem{lemma}[theorem]{Lemma}
    \newtheorem{defn}[theorem]{Definition}
    \newtheorem{remark}[theorem]{Remark}

    \newtheorem{ass}[theorem]{Assumption}
    \crefname{ass}{Assumption}{Assumptions}
    
    \newtheorem{prop}[theorem]{Proposition}
    \crefname{prop}{Proposition}{Propositions}
    
    \newcommand{\bbR}{\mathbb{R}}
    \newcommand{\cL}{\mathcal{L}}
    \newcommand{\cC}{\mathcal{C}}
    \newcommand{\cH}{\mathcal{H}}
    
    \newcommand{\bff}{\boldsymbol{f}}
    \newcommand{\bfg}{\boldsymbol{g}}
    \newcommand{\bfu}{\boldsymbol{u}}
    \newcommand{\bfv}{\boldsymbol{v}}
    
    \newcommand{\diff}{\,\textnormal{d}}
    \newcommand{\dt}{\Delta t}
    \newcommand{\dv}{\,\diff\bfv}
    \newcommand{\intd}{{\int_{\bbR^d}}}
    \newcommand{\mfn}{\mathfrak{n}}
    
    \newcommand{\vint}[1]{\langle \!\; #1 \!\; \rangle}
    \newcommand{\Vint}[1]{\left\langle #1 \right\rangle}

    \newcommand{\qquand}{\qquad \mbox{and} \qquad}
    
    \newcommand{\qqwhere}{\qquad \mbox{where} \qquad}
    
    \newcommand{\quiff}{\quad \iff \quad}

    \allowdisplaybreaks

    \date{17 April 2024}

\title{A Nonlinear, Conservative, Entropic Fokker-Planck Model for Multi-Species Collisions
\thanks{
  This manuscript has been authored, in part, by UT-Battelle, LLC, under contract DE-AC05-00OR22725 with the US Department of Energy (DOE). The U.S. government retains and the publisher, by accepting the article for publication, acknowledges that the U.S. government retains a nonexclusive, paid-up, irrevocable, worldwide license to publish or reproduce the published form of this manuscript, or allow others to do so, for U.S. government purposes. DOE will provide public access to these results of federally sponsored research in accordance with the DOE Public Access Plan (https://energy.gov/downloads/doe-public-access-plan).
  \\
  This work was supported, in part, by the U.S. Department of Energy through the Los Alamos National Laboratory. Los Alamos National Laboratory is operated by Triad National Security, LLC, for the National Nuclear Security Administration of U.S. Department of Energy (Contract No. 89233218CNA000001).
  \\
  J. Hu’s research is partially supported by AFOSR grant FA9550-21-1-0358, and DOE grant DE-SC0023164.
  }
}

\author{
Evan Habbershaw\thanks{
    Department of Mathematics, The University of Tennessee, Knoxville, TN, 37996
  }, 
Cory Hauck\thanks{
    Department of Mathematics, The University of Tennessee, Knoxville, TN, 37996
    and
    Computer Science and Mathematics Division, Oak Ridge National Laboratory, Oak Ridge, TN, 37831
  }, 
Jingwei Hu\thanks{
  Department of Applied Mathematics, University of Washington, Seattle, WA, 98195
  }, and 
Jeffrey R. Haack\thanks{
  Computer, Computational, and Statistical Sciences Division, Los Alamos National Laboratory, Los Alamos, NM 87545
  }
}

\begin{document}
    \pagenumbering{arabic}
    
    \maketitle

    \begin{abstract}
\noindent A multi-species Fokker-Planck model for simulating particle collisions in a plasma is presented.  The model includes various parameters that must be tuned.  Under reasonable assumptions on these parameters, the model satisfies appropriate conservation laws, dissipates an entropy, and satisfies an $\mathcal{H}$-Theorem.  In addition, the model parameters provide the additional flexibility that is used to match simultaneously momentum and temperature relaxation formulas derived from the Boltzmann collision operator for a binary mixture with Coulomb potentials.   A numerical method for solving the resulting space-homogeneous kinetic equation is presented and two examples are provided to demonstrate the relaxation of species bulk velocities and temperatures to their equilibrium values.
    \end{abstract}
    
\noindent\textbf{Keywords:} Multi-species kinetic models, Fokker-Planck, Lenard-Bernstein, Dougherty, moment equations, entropy, conservation laws, plasma physics
    
    \section{Introduction}
    
In this paper, we propose and analyze a multi-species version of a kinetic collision model for plasma physics applications, often referred to as the Lenard-Bernstein \cite{lenard1958plasma} or Dougherty model \cite{dougherty1964model}.  
The single-species version of this model (which we refer to as LB) is a nonlinear Fokker-Planck operator that is used as a computationally cheaper surrogate for the full Landau-Fokker-Planck (LFP) operator \cite{landau1958kinetic}.  
Roughly speaking, the LB model replaces the nonlocal advection and diffusion terms of the LFP  operator by parameters that depend only on the bulk velocity and temperature of the underlying kinetic distribution \cite{liboff2003kinetic}, while maintaining the conservation, entropy dissipation, and equilibrium properties of the original LFP operator.  
In this way, the approximation of LFP by LB is similar in spirit to the approximation of the Boltzmann collision operator (see e.g, \cite{cercignani2013mathematical}) by the Bhatnagar-Gross-Krook (BGK) model \cite{bhatnagar1954model}.

Recently, there have been efforts to extend the LB operator to the multi-species case \cite{ulbl2022implementation,francisquez2022improved} (see \cite{HJZ24} for the well-posedness result of a similar model) in a manner that maintains relevant (global) conservation properties, entropy dissipation, and an $\cH$-Theorem that uniquely characterizes equilibrium states.  
These multi-species models are written as a sum of Fokker-Planck operators, with coefficients that can be interpreted as tunable collision frequencies.  
Each Fokker-Planck term is parameterized by species bulk velocity and temperature in the case of intra-species collisions or by mixture values for the bulk velocity and temperature in the case of inter-species collisions.  
The choice of mixture values uses the strategy for the multi-species BGK (M-BGK) model introduced in \cite{Haack2017}, which preserves the aforementioned structural properties.  
However, like \cite{Haack2017}, the models in \cite{ulbl2022implementation,francisquez2022improved}, can only match the pairwise relaxation rate (derived from the Boltzmann collision operator for a binary mixture with Coulomb potential)
\footnote{This matching procedure is relevant since for small angle collisions, the Boltzmann collision operator for a binary mixture with Coulomb potential is approximated by the LFP operator \cite{GAMBA201440}.}
of either the bulk velocity or temperature, but not both at the same time.  
This limitation is a consequence of the fact that mixture values for the bulk velocity and temperature are assumed to be symmetric with respect to a permutation in indices between two species.

A more general M-BGK model was introduced in \cite{klingenberg2017consistent} (see also \cite{bobylev2018general}) that does not make the symmetry assumption in \cite{Haack2017}.
Rather, the model in \cite{klingenberg2017consistent}, which includes the one in \cite{Haack2017} as a special case, introduces additional parameters in the mixture values to enable matching of pairwise bulk velocity and temperature relaxation rates at the same time.  
In this work, we take a similar approach and construct a more general multi-species LB (M-LB) model that satisfies the requisite conservation, entropy dissipation, and equilibrium ($\cH$-Theorem) properties, but contains additional parameters for matching of velocity and temperature relaxation rates at the same time.
The parameters however, have constraints; thus a principled approach for matching must be carried out.  
We develop a procedure for doing so, matching the rates of the M-LB model with those derived from the Boltzmann collision operator with Coulomb potentials.  
We then illustrate the application of the matching procedure with several numerical experiments.

The remainder of the paper is organized as follows. 
In \Cref{section:MLB}, we introduce the M-LB model.   
In \Cref{section:conservationLaws}, we review conservation laws and derive required physically relevant conditions on the mixture parameters for the bulk velocity and temperature.  
In \Cref{section:HTheorem}, we prove an entropy dissipation law and an $\cH$-Theorem.  
In \Cref{section:matchingRelaxationRates}, we describe the procedure for matching the pairwise relaxation rates to the multi-species Boltzmann equation with Coulomb potentials.  
In \Cref{section:numerics}, we present an implicit method for solving the homogeneous M-LB equations, which includes an iterative procedure for updating the moments when presented with moment-dependent collision rates.   
In \Cref{section:numericalExamples}, we present some preliminary numerical results for the space homogeneous problem.   
In the appendix we present proofs used in the main body of the paper, and provide comparisons to the recent work in \cite{pirner2024consistent}.

    \section{Multi-species space homogeneous Lenard-Bernstein equation}
    \label{section:MLB}
    
Let $f_i(t,\bfv)$ denote the kinetic distribution of particles of species $i\in\{1,\cdots,N\}$ having mass $m_i$.
Associated to each $f_i$ are the species number density $n_i$, mass density $\rho_i$, bulk velocity $\bfu_i$, and temperature $T_i$, defined by
    \begin{equation}
\rho_i = m_in_i = m_i\intd f_i\dv
    ,\;\;\;
\bfu_i = \frac{1}{\rho_i}\intd m_i\bfv f_i\dv
    ,\;\;\;
T_i = \frac{m_i}{n_i d}\intd\left|\bfv-\bfu_i\right|^2f_i\dv
.
    \end{equation}
The species energy densities, $E_i$, are given by
    \begin{equation}
    \label{eq:energyDefinition}
E_i
  = \frac{1}{2}\intd m_i|\bfv|^2f_i\dv
  = \frac{1}{2}\rho_i|\bfu_i|^2 + \frac{d}{2}n_i T_i
.
    \end{equation}
    
    \begin{defn}
For $n>0$, $\bfu \in \bbR^d$, and $\theta>0$, define the Maxwellian
    \begin{equation}
    \label{eq:maxwellianDefinition}
M_{n,\bfu,\theta}
  = 
    \frac{n}{(2\pi\theta)^{d/2}}
    \exp
    \left(
      -\frac{|\bfv-\bfu|^2}{2\theta}
    \right)
.
    \end{equation}
    \end{defn}
    
    \begin{lemma}
For a general Maxwellian $M_{n,\bfu,\theta}$,
    \begin{gather}
\intd M_{n,\bfu,\theta}\dv = n,
    \qquad 
\intd \bfv M_{n,\bfu,\theta}\dv = n\bfu, 
    \\
    \label{eq:DeltaM_identity}
\intd |\bfv|^2M_{n,\bfu,\theta} \dv = n|\bfu|^2+dn\theta,
\qquand
\Delta_{\bfv}\log(M_{n,\bfu,\theta})
  =
    -\frac{d}{\theta}
.
    \end{gather}
    \end{lemma}
    
    \begin{defn}
The (space homogeneous) multi-species Lenard Bernstein (M-LB) model is 
    \begin{equation}
    \label{eq:MLB}
\frac{\partial f_i}{\partial t}
  = 
    \sum_{j=1}^N \cC_{i,j}(f_i,f_j)
,\qquad i=1,\cdots,N
,
    \end{equation}
where, for $1 \leq i,j \leq N$, the pairwise collision operators $\cC_{i,j}$ are given by
    \begin{subequations}
    \label{eq:MLB_operator}
    \begin{align}
\cC_{i,j}(f_i,f_j)
  &\coloneqq
    \lambda_{i,j}\theta_{i,j}\nabla_{\bfv}\cdot
      \left(
        \frac{\bfv-\bfu_{i,j}}{\theta_{i,j}}f_i
        +
        \nabla_{\bfv} f_i
      \right)
  = 
    \lambda_{i,j} \theta_{i,j} \cL(f_i, M_{i,j})
,
   \\
\cL(f,M)
  &\coloneqq
    \nabla_{\bfv}\cdot
      \left(
        M \nabla_{\bfv} 
        \left(
          \frac{f}{M}
        \right)
      \right)
  = 
    \nabla_{\bfv}\cdot
      \left[
        \nabla_{\bfv} f - f \nabla_{\bfv} (\log(M)) 
      \right]
.
    \end{align}    
    \end{subequations}
The parameters $\lambda_{i,j} \geq 0$ are the frequencies of collisions between species $i$ and $j$.  
The mixture Maxwellians are given by
    \begin{equation}
    \label{eq:mixtureMaxwellian_and_theta_ij_def}
M_{i,j}
  \coloneqq
    M_{n_i,\bfu_{i,j},\theta_{i,j}}
,\qqwhere
\theta_{i,j}
  \coloneqq
    T_{i,j}/m_i
,
    \end{equation}
and the mixture Maxwellian parameters are given by
    \begin{subequations}
    \label{eq:u_ij_and_T_ij_defs}
    \begin{align}
    \label{eq:u_ij_def}
\bfu_{i,j}
  &\coloneqq
    \alpha_{i,j}\bfu_i+(1-\alpha_{i,j})\bfu_j
,
    \\
    \label{eq:T_ij_def}
T_{i,j}
  &\coloneqq
    \beta_{i,j}T_i+(1-\beta_{i,j})T_j
    +\frac{1}{d}\gamma_{i,j}|\bfu_i-\bfu_j|^2
,
    \end{align}    
    \end{subequations}
where the parameters $\alpha_{i,j}$, $\beta_{i,j}$, and $\gamma_{i,j}$ are to be determined.
    \end{defn}
    
    \begin{remark}
When $i = j$, we often drop the second index.
Thus, for example, $M_{i,i} = M_i$ is the unique Maxwellian associated to $f_i$ with bulk velocity $\bfu_{i,i} = \bfu_i$ and temperature $T_{i,i} = T_i$.
    \end{remark}

    \section{Conservation laws and mixture parameter constraints}
    \label{section:conservationLaws}
    
Throughout the paper, we make the following assumption, which allows us to use the chain rule for time derivatives, apply integration by parts in velocity, and, when doing the latter, set any boundary terms that arise to zero.
While more precise conditions can be formulated, we do not do so here.

    \begin{ass}
For each $i\in\{1,\cdots,N\}$, $f_i>0$ is a smooth function of $\bfv$ and continuously differentiable in $t$.
For each $t$, $f_i$ decays exponentially to zero as $|\bfv| \to \infty$.
    \end{ass}

The M-LB model satisfies conservation of species mass:  $m_i\intd\cC_{i,j}(f_i,f_j)\dv=0$.
We assume further that it satisfies the pairwise conservation of  momentum and energy.

    \begin{ass}
    \label{ass:conservation}
The collision operators in \eqref{eq:MLB} satisfy the following invariance properties, which correspond to pairwise momentum conservation and pairwise energy conservation, respectively.
Specifically, for any $i,j$,
    \begin{subequations}
    \begin{align}
    \label{eq:cLaw2}
m_i\vint{\bfv\cC_{i,j}(f_i,f_j)}
+
m_j\vint{\bfv\cC_{j,i}(f_j,f_i)}
    &= \mathbf{0}
,
    \\ 
    \label{eq:cLaw3}
m_i\Vint{\frac{|\bfv|^2}{2}\cC_{i,j}(f_i,f_j)}
+
m_j\Vint{\frac{|\bfv|^2}{2}\cC_{j,i}(f_j,f_i)}
    &= 0
,
    \end{align}
    \end{subequations} 
where brackets are used as a shorthand for velocity integration, i.e.,
$\vint{\bfg} = \intd {\bfg}({\bfv})\diff{\bfv}$.
    \end{ass}
Evaluation of these conservation laws relies on the following lemma.
    \begin{lemma}
    \label{lemma:momentProperties}
For every $i,j\in\{1,\cdots,N\}$,
    \begin{subequations}
    \label{eq:momentProperties}
    \begin{align}
    \label{eq:momentProperty1}
\vint
  {m_i\cL(f_i,M_{i,j})}
  &= 0
,
    \\
    \label{eq:momentProperty2}
\vint
  {m_i\bfv\cL (f_i,M_{i,j})}
  &= 
    -\frac
        {\rho_i(\bfu_i - \bfu_{i,j})}
        {\theta_{i,j}}
,
    \\
    \label{eq:momentProperty3}
\Vint
  {m_i\frac{|\bfv|^2}{2} \cL(f_i,M_{i,j})}
  &= 
    -\frac
      {\rho_i \bfu_i \cdot (\bfu_i - \bfu_{i,j})}
      {\theta_{i,j}}
    -\frac
      {d \rho_i (\theta_i - \theta_{i,j})}
      {\theta_{i,j}}
.
    \end{align}
    \end{subequations}
    \end{lemma}
    
    \begin{proof}
The proof is a straightforward, but tedious computation, using integration by parts several times.
    \end{proof}

    \begin{prop}
    \label{prop:parameters_ji}
Given the definitions of $\bfu_{i,j}$ and $T_{i,j}$ in \eqref{eq:u_ij_and_T_ij_defs}, the pairwise conservation of momentum and energy in \Cref{ass:conservation} holds if and only if the following parameter relationships hold:
    \begin{subequations}
    \label{eq:ji_parameter_defs}
    \begin{align}
\alpha_{j,i}
  &=
    1-\frac{\rho_i\lambda_{i,j}}{\rho_j\lambda_{j,i}}(1-\alpha_{i,j})
,
    \\
    \label{eq:beta_ji_def}
\beta_{j,i}
  &=
    1-\frac{n_i\lambda_{i,j}}{n_j\lambda_{j,i}}(1-\beta_{i,j})
,
    \\
    \label{eq:gamma_ji_def}
\gamma_{j,i}
  &=
    \frac{1}{n_j\lambda_{j,i}}(\rho_i\lambda_{i,j}(1-\alpha_{i,j})-n_i\lambda_{i,j}\gamma_{i,j})
.
    \end{align}
    \end{subequations}
    \end{prop}
    
    \begin{proof}
We show that \Cref{ass:conservation} implies \eqref{eq:ji_parameter_defs}.
The converse follows by simply reversing the order of computations.
According to \eqref{eq:momentProperty2}, the conservation of momentum \eqref{eq:cLaw2} implies that 
    \begin{equation}
\rho_i\lambda_{i,j} \bfu_{i,j} + \rho_j\lambda_{j,i} \bfu_{j,i}
  =
    \rho_i\lambda_{i,j}\bfu_{i} + \rho_j\lambda_{j,i}\bfu_{j}
.
    \end{equation}
Rearranging this equation to isolate $\bfu_{j,i}$ and applying \eqref{eq:u_ij_def} yields \eqref{eq:beta_ji_def}.
By \eqref{eq:momentProperty3}, the conservation of energy \eqref{eq:cLaw3} implies that
    \begin{equation}
n_i\lambda_{i,j}T_{i,j} + n_j\lambda_{j,i}T_{j,i} 
  = 
    \frac{1}{d}
    \rho_i\lambda_{i,j}(1-\alpha_{i,j})
      |\bfu_i - \bfu_j|^2
    + n_i\lambda_{i,j}T_{i}
    + n_j\lambda_{j,i}T_{j}
.
    \end{equation}
Rearranging this equation to isolate $T_{j,i}$ and using \eqref{eq:T_ij_def} yields \eqref{eq:beta_ji_def} and \eqref{eq:gamma_ji_def}.
    \end{proof}
    
    \begin{remark}
The relations in \eqref{eq:ji_parameter_defs} can be rearranged as follows, in order to highlight the symmetries in the model:
    \begin{subequations}
    \label{eq:modelSymmetries}
    \begin{align}
    \label{eq:deltaSymmetry}
\delta_{i,j}
  \coloneqq
    \rho_i\lambda_{i,j}(1-\alpha_{i,j})
  &=
    \rho_j\lambda_{j,i}(1-\alpha_{j,i})
  \eqqcolon
    \delta_{j,i},
    \\
    \label{eq:betaSymmetry}
n_i\lambda_{i,j}(1-\beta_{i,j})
  &=
    n_j\lambda_{j,i}(1-\beta_{j,i}),
    \\
    \label{eq:gammaSymmetry}
\delta_{i,j} - 2n_i\lambda_{i,j}\gamma_{i,j}
  &= 
    2n_j\lambda_{j,i}\gamma_{j,i} - \delta_{j,i}
.
    \end{align}
    \end{subequations}
    \end{remark}

    \begin{remark}
If we constrain the mixture velocity terms to enforce the symmetry $\bfu_{i,j} = \bfu_{j,i}$, then $\alpha_{i,j}+\alpha_{j,i}=1$, and 
    \begin{equation}
\bfu_{i,j}
  = 
    \alpha_{i,j}\bfu_i+\alpha_{j,i}\bfu_j
,\qqwhere
\alpha_{i,j}
  =
    \frac{\rho_i\lambda_{i,j}}{\rho_i\lambda_{i,j}  + \rho_j\lambda_{j,i}}
.
    \end{equation}
This is the same expression for the mixture velocity as in \cite{Haack2017}.
If, in addition, we require $T_{i,j} = T_{j,i}$, then $\beta_{i,j}+\beta_{j,i}=1$, $\gamma_{i,j}=\gamma_{j,i}$, and
    \begin{equation}
T_{i,j}
  = 
    \beta_{i,j}T_i
    +
    \beta_{j,i}T_j
    +
    \frac{1}{d}
    \frac
      {\delta_{i,j} |\bfu_i - \bfu_j|^2}
      {\lambda_{i,j} n_i + \lambda_{j,i} n_j}
,    
    \end{equation}
where
    \begin{equation}
\beta_{i,j}
  = 
    \frac{n_i\lambda_{i,j}}{n_i\lambda_{i,j}+n_j\lambda_{j,i}}
\qquand
\delta_{i,j}
  =
    \frac
        {\rho_{i}\rho_{j} \lambda_{i,j} \lambda_{j,i}}
        {\rho_i\lambda_{i,j} + \rho_j\lambda_{j,i}}
.
    \end{equation}
In this case
    \begin{align}
T_{i,j} 
  &= 
    \frac
        {n_i\lambda_{i,j}T_i + n_j\lambda_{j,i}T_j}
        {n_i\lambda_{i,j} + n_j\lambda_{j,i}}
    +
    \frac{1}{d}
    \frac
        { \rho_i\rho_j\lambda_{i,j}\lambda_{j,i}}
        {(\rho_i\lambda_{i,j} + \rho_j\lambda_{i,j})(n_i\lambda_{i,j} + n_j\lambda_{j,i})}
    |\bfu_i - \bfu_j|^2
,
    \end{align}
which is the mixture temperature in \cite{Haack2017}.  
    \end{remark}

The following moment ODE system, describing the time evolution of the momenta, and energies can be derived (see \Cref{appendix:modelRelaxationRates}) from the model equation \eqref{eq:MLB}, by integrating against $m_i(\bfv,|\bfv|/2)^\top$, and using \eqref{eq:momentProperties} and the definitions in \eqref{eq:u_ij_and_T_ij_defs}:
    \begin{subequations}
    \label{eq:momentODEsystem}
    \begin{align}
    \label{eq:momentumODE}
\frac{\partial(\rho_i\bfu_i)}{\partial t}
  &=
    \sum_j\rho_i\lambda_{i,j}(1-\alpha_{i,j})(\bfu_j-\bfu_i)
,
    \\
    \label{eq:energyODE}
\frac{\partial E_i}{\partial t}
  &=
    d\sum_jn_i\lambda_{i,j}(1-\beta_{i,j})(T_j-T_i)
    \\
  &\quad+
\!\sum_j
      \bigg[
        (\delta_{i,j}-n_i\lambda_{i,j}\gamma_{i,j})
        \bfu_i\cdot(\bfu_j-\bfu_i)
        -
        (\delta_{j,i}-n_j\lambda_{j,i}\gamma_{j,i})
        \bfu_j\cdot(\bfu_i-\bfu_j)
      \bigg]
.
    \nonumber
    \end{align}
    \end{subequations}
Using \eqref{eq:momentProperties} and \eqref{eq:momentumODE}, the following ODE describing the temperature can be derived:
    \begin{equation}
\frac{d}{2}
\frac{\partial(n_iT_i)}{\partial t}
  =
    d\sum_jn_i\lambda_{i,j}(1-\beta_{i,j})(T_j-T_i) 
    +
    \sum_jn_i\lambda_{i,j}\gamma_{i,j}|\bfu_i-\bfu_j|^2
.
    \end{equation}

    \begin{lemma}
The ODE system \eqref{eq:momentODEsystem} satisfies the conservation of total momentum and total energy:
    \begin{equation}
    \label{eq:cons_total_momentum_and_energy}
\sum_i\frac{\partial(\rho_i\bfu_i)}{\partial t}
  = \mathbf{0}
\qquand
\sum_i\frac{\partial E_i}{\partial t}  
  = 0
,
    \end{equation}
and the quantities
    \begin{equation}
    \label{eq:u_inf_and_T_inf_defs}
\bfu^\infty
  = \frac{\sum_i\rho_i\bfu_i}{\sum_i\rho_i}\in\bbR^d
\qquand
T^\infty
  = 
    \frac{\sum_in_iT_i}{\sum_in_i}
    +
    \frac{
        \sum_i\rho_i
        \left(
          |\bfu_i|^2-|\bfu^\infty|^2
        \right)
        }
        {d\sum_in_i}
    >0
    \end{equation}
are independent of time.
    \end{lemma}
    
    \begin{proof}
To verify the conservation of total momentum, sum \eqref{eq:momentumODE} over all species $i\in\{1,\cdots,N\}$ and use the symmetry in \eqref{eq:deltaSymmetry} to obtain
    \begin{equation}
\sum_i\frac{\partial(\rho_i\bfu_i)}{\partial t}
  = 
    \sum_i\sum_j\rho_i\lambda_{i,j}(1-\alpha_{i,j})(\bfu_j-\bfu_i)
  = 
    \mathbf{0}
.
    \end{equation}
Next, to verify the conservation of total energy, sum \eqref{eq:energyODE} over $i\in\{1,\cdots,N\}$ and use \eqref{eq:betaSymmetry}.
Finally, the time invariance of $\bfu^\infty$ and $T^\infty$ follow directly from \eqref{eq:cons_total_momentum_and_energy}, and the positivity of $T_\infty$ follows as in \cite[Proposition 4.1]{asymptoticRelaxation}.
    \end{proof}
    
    \begin{remark}
Following arguments similar, to that of \cite[Theorem 3.2]{asymptoticRelaxation} one can show that the temperatures $T_i$, $i \in \{1,\cdots,N\}$, are bounded below and that $min_{1\leq i \leq N} {T_i}$ is non-decreasing.
Using this fact, the global existence and uniqueness of solutions to the moment system and the convergence for each $i \in \{1,\cdots,N\}$ of $\bfu_i \to \bfu^\infty$ and $T_i \to T^\infty$ as $t \to \infty$ can be established, in a way similar to \cite{asymptoticRelaxation}.
    \end{remark}

    \section{Entropy dissipation and \texorpdfstring{$\cH$}{H}-Theorem}
    \label{section:HTheorem}
    
In this section, we state and prove an entropy dissipation result and an $\cH$-Theorem, under the following physically motivated assumption.
    \begin{ass}
    \label{ass1}
Motivated by physical principles, we assume that for $1 \leq i,j \leq N$,
    \begin{equation}
    \label{eq:ij_parameter_constraints}
0 \leq \alpha_{i,j}\,,\,\beta_{i,j} \leq 1
\qquand
\gamma_{i,j} \geq 0
.
    \end{equation}
    \end{ass}

    \begin{defn}
The total entropy functional for a particle system with kinetic distribution $\bff = [f_1, \cdots f_N]^\top$ is given by 
    \begin{equation}
    \label{eq:Hdefn}
\cH[\bff] 
  = 
    \sum_{i=1}^N\intd (f_i\log(f_i) -f_i)\dv 
  = 
    \sum_{i=1}^N\vint{f_i\log(f_i) - f_i}
.
    \end{equation} 
    \end{defn}
    
    \begin{theorem}
    \label{theorem:HTheorem}
Let $\bff = [f_1, \cdots f_N]^\top$ satisfy the M-LB system
\eqref{eq:MLB}-\eqref{eq:MLB_operator}.
Then under the assumptions in \eqref{eq:ij_parameter_constraints}, $\frac{\partial\cH}{\partial t}\leq 0$, with equality if and only if each $f_i$ is a Maxwellian with common equilibrium velocity $\bfu\in\bbR^d$ and equilibrium temperature $T>0$; that is,
    \begin{align}
    \label{eq:equilibrium}
f_i(\bfv)
  &=
    n_i
    \left(
      \frac{m_i}{2\pi T^\infty}
    \right)^{d/2}
    \exp
    \left(
      -
      \frac{m_i|\bfv-\bfu^\infty|^2}{2T^\infty}
    \right)
,\qquad 1 \leq i \leq N
,
    \end{align}
where $\bfu^\infty$ and $T^\infty$ are defined in \eqref{eq:u_inf_and_T_inf_defs}.
    \end{theorem}

    \begin{proof}
By the definition of $\cH$ in \eqref{eq:Hdefn} and the chain rule,
    \begin{equation}
    \begin{split}
\frac{\partial \cH}{\partial t} 
  &= 
    \sum_{i=1}^N \vint{\log(f_i) \frac{\partial f_i}{\partial t}}
    \\
  &=
    \sum_{i=1}^N \vint{\log(f_i) \cC_{i,i}(f_i,f_i)}
      +
    \sum_{i=1}^N \sum_{j=i+1}^N 
      \vint{\log(f_i) \cC_{i,j}(f_i,f_j) + \log(f_j) \cC_{j,i}(f_j,f_i)}
.
   \end{split}
   \end{equation}
Thus it is sufficient to show that 
    \begin{equation}
\vint{\log(f_i) \cC_{i,i}(f_i,f_i)} 
  \leq 0
,
    \end{equation}
with equality if and only if $f_i$ takes the form in \eqref{eq:equilibrium}, and
    \begin{equation}
    \label{eq:pairwiseEntropy}
\vint
    {\log(f_i) \cC_{i,j}(f_i,f_j) + \log(f_j) \cC_{j,i}(f_j,f_i)} 
  \leq 0, 
    \end{equation}
with equality if and only if $f_i$ and $f_j$ take the Maxwellian form in \eqref{eq:equilibrium} with a common bulk velocity $\bfu$ and temperature $T>0$.
We establish these conditions in the remainder of the section.
    \end{proof}

    \begin{lemma}
    \label{lemma:singleSpeciesEntropyIneq}
For the intra-species collisions,
    \begin{equation}
    \label{eq:intra-species_Hthm}
\vint
    { \log (f_i) \cC_{i,i} (f_i,f_i) }
  \leq 0
,
    \end{equation}
with equality if and only if $f_i = M_{i}$.
    \end{lemma}
    
    \begin{proof}
Since
    \begin{equation}
\log(M_{i})
  = 
    \log
      \left(
        \frac{n_i}{(2\pi\theta_i)^{d/2}}
      \right)
      -
      \frac{|\bfv-\bfu_i|^2}{2\theta_i}
,
    \end{equation}
it follows from \Cref{lemma:momentProperties} that 
    \begin{equation}
\Vint{\log(M_{i}) \cC(f_i,f_i)} = 0
.
    \end{equation}
Hence, after integrating by parts,
    \begin{equation}
    \begin{split}
\Vint{\log(f_i) \cC_{i,i}(f_i,f_i)}
  &= 
    \Vint{\log \left( \frac{f_i}{M_{i}} \right) \cC_{i,i}(f_i,f_i)}
    \\
  =
  \lambda_{i,i} \theta_i 
    &\Vint
      {
        \log
        \left(
          \frac{f_i}{M_i}
        \right)
        \nabla_{\bfv}\cdot
        \left(
          M_i\nabla_{\bfv}
          \left(
            \frac{f}{M_i}
          \right)
        \right)
      }
  =
    -\lambda_{i,i}\theta_i
    \Vint
      {
        \frac{M_i^2}{f}
        \left|
          \nabla_{\bfv}
          \left(
            \frac{f_i}{M_i}
          \right)
        \right|^2
      }
.
    \end{split}
    \end{equation}
Thus the inequality in \eqref{eq:intra-species_Hthm} holds, with equality if and only if $f_i/M_i$ is constant with respect to $\bfv$.
Since $\vint{f_i} = \vint{M_i}$, this constant must be one, and the result follows.
    \end{proof}
    
We now prove \eqref{eq:pairwiseEntropy}.
Several intermediate steps are required.

    \begin{lemma}
    \label{lemma:genericMaxwellianIdentity}
For a general Maxwellian $M_{n,\bfu,\theta}$,
    \begin{equation}
\vint{
  \log(f) \cL(f,M_{n,\bfu,\theta})
}
  =
    -\vint
      {
        f \Delta_{\bfv}
        \left(
          \log(f) - \log(M_{n,\bfu,\theta})
        \right)
      }
,
    \end{equation}
which implies that
    \begin{equation}
    \label{eq:identity1}
\vint{\log(f_i) \cL(f_i,M_{i,j})}
  - \vint{\log(f_i) \cL(f_i,M_{i,i})}
=
\vint{f_i \Delta_{\bfv} \left( \log(M_{i,i}) - \log(M_{i,j}) \right)}
.
    \end{equation}
    \end{lemma}

    \begin{proof}
Using integration by parts,
    \begin{equation}
    \begin{split}
\vint{
  \log(f) \cL(f,M)
}
  &= \vint{\log(f) \nabla_{\bfv}  \cdot
    \left[
      \nabla_{\bfv} f - f \nabla_{\bfv} (\log(M)) 
    \right]}
    \\
  &= -\vint{ \nabla_{\bfv} \log(f) \cdot  \nabla_{\bfv} f } 
    + \vint{f \nabla_{\bfv} \log(f) \cdot \nabla_{\bfv} \log(M) }
    \\
  &= 
    -\vint
      {\nabla_{\bfv} \log(f) \cdot  \nabla_{\bfv} f} 
    +
    \vint
      {\nabla_{\bfv} f \cdot \nabla_{\bfv} \log(M)}
    \\
  &= \vint{\nabla_{\bfv} f \cdot \nabla_{\bfv}
    \left(
      \log(M) - \log(f)
    \right)}
    \\
  &= -\vint{f \Delta_{\bfv} \left( \log(f) - \log(M) \right)}
.
    \end{split}
    \end{equation}
Identity \eqref{eq:identity1} follows immediately.
    \end{proof}

The key point of \Cref{lemma:genericMaxwellianIdentity} is that the entropy dissipation due to $\cL(f_i,M_{i,j})$ can be expressed in terms of moments of $f_i$ and $f_j$.
This fact is used to write the inter-species entropy dissipation in terms of the intra-species entropy dissipation and some moment-dependent corrections.

    \begin{lemma}
    \label{lemma:mixed_identity}
    \begin{align}
    \label{eq:mixed_identity}
\vint{\log(f_i)\cC_{i,j}(f_i,f_j)}
  &=
    \frac{\lambda_{i,j}\theta_{i,j}}{\lambda_{i,i}\theta_{i,i}}\vint{ \log(f_i)\cC_{i,i} (f_i,f_i)} 
    +
    dn_i\lambda_{i,j}
    \frac{T_{i}-T_{i,j}}{T_{i}}
.
    \end{align}
    \end{lemma}

    \begin{proof}
Adding and subtracting $\lambda_{i,j}\theta_{i,j} \vint{ \log(f_i) \cL(f_i,M_{i,j}) }$ on the left-hand side of \eqref{eq:mixed_identity} and then applying \Cref{lemma:genericMaxwellianIdentity} gives
    \begin{align}
\vint{ \log(f_i)\cC_{i,j} (f_i,f_j)} 
&= \lambda_{i,j}  \theta_{i,j} \vint{ \log(f_i) \cL(f_i,M_{i,j}) }
    \nonumber
    \\
 &=\lambda_{i,j}
   \theta_{i,j}\!\Vint{ \log(f_i) \cL(f_i,M_{i,i}) }
    \!+\!
    \lambda_{i,j}\theta_{i,j} \!\Vint{ \log(f_i) \cL(f_i,M_{i,j}) }
    \!-\!
    \lambda_{i,j}  \theta_{i,j} \!\Vint{ \log(f_i) \cL(f_i,M_{i,i}) }
    \nonumber
    \\
  &=\frac{\lambda_{i,j}\theta_{i,j}}{\lambda_{i,i}\theta_{i,i}}\vint{ \log(f_i)\cC_{i,i}(f_i,f_i) } 
    + \lambda_{i,j}  \theta_{i,j}
    \Vint{f_i\Delta_{\bfv} \left[\log(M_{i,i}) - \log(M_{i,j})\right]}
    \nonumber
    \\
  &=\frac{\lambda_{i,j}\theta_{i,j}}{\lambda_{i,i}\theta_{i,i}}\vint{ \log(f_i)\cC_{i,i}(f_i,f_i) } 
    + dn_i\lambda_{i,j}
    \frac{T_{i}-T_{i,j}}{T_{i}}
,
    \end{align}
where in the final line above we have used \eqref{eq:DeltaM_identity} and \eqref{eq:mixtureMaxwellian_and_theta_ij_def}.
    \end{proof}

    \begin{lemma}
The cross terms in \eqref{eq:pairwiseEntropy} satisfy
    \begin{align}
\vint{\log(f_i) \cC_{i,j}(f_i,f_j)} + \vint{\log(f_j) \cC_{j,i}(f_j,f_i)}
  \leq 0
,
    \end{align}
with equality if and only if $f_i$ and $f_j$ are Maxwellians with a common bulk velocity and temperature.
    \end{lemma}
    
    \begin{proof}
Using
\eqref{eq:T_ij_def}, \eqref{eq:betaSymmetry}, and \Cref{lemma:mixed_identity},
    \begin{equation}
    \label{eq:lemmaProofEquation}
    \begin{split}
\vint{\log(f_i)\cC_{i,j}(f_i,f_j)}
  +
\vint{\log(f_j)\cC_{j,i}(f_j,f_i)}
   &=
    \frac
        {\lambda_{i,j}\theta_{i,j}}
        {\lambda_{i,i}\theta_{i,i}}
    \vint{\log(f_i)\cC_{i,i}(f_i,f_i)}
    +
    \frac
        {\lambda_{j,i}\theta_{j,i}}
        {\lambda_{j,j}\theta_{j,j}}
    \vint{\log(f_j)\cC_{j,j}(f_j,f_j)}
    \\
-
    dn_i\lambda_{i,j}(1-\beta_{i,j})
    &\frac{(T_i-T_j)^2}{T_iT_j}
  -
    \left(
      n_i\lambda_{i,j}\gamma_{i,j}T_j+n_j\lambda_{j,i}\gamma_{j,i}T_i
    \right)
    \frac{|\bfu_i-\bfu_j|^2}{T_iT_j}
.
    \end{split}
    \end{equation} 
By \Cref{lemma:singleSpeciesEntropyIneq}, the right hand side of \eqref{eq:lemmaProofEquation} is non-positive.
Moreover, it is equal to zero if and only if $f_i$ and $f_j$ are Maxwellians with $\bfu_i = \bfu_j$ and $T_i=T_j$.  
    \end{proof}
    
    \begin{remark}
While the condition on $\alpha_{i,j} \in [0,1]$ is not explicitly used in the proofs above, the condition that $\gamma_{i,j}\geq0$, which is used, may fail to hold if the bound on $\alpha_{i,j}$ does not hold.
In particular, \eqref{eq:deltaSymmetry} and   \eqref{eq:gammaSymmetry} implies that
    \begin{equation}
\rho_i\lambda_{i,j}(1-\alpha_{i,j})
  =
    n_i\lambda_{i,j}\gamma_{i,j}
    +
    n_j\lambda_{j,i}\gamma_{j,i}
.
    \end{equation}
Thus if $\alpha_{i,j}>1$, either $\gamma_{i,j}$ or $\gamma_{j,i}$ will fail to be positive.
    \end{remark}

    \section{Matching relaxation rates}
    \label{section:matchingRelaxationRates}
    
In this section, momentum and energy relaxation rates for the current model are matched to the relaxation rates for the Boltzmann collision operator with Coulomb potential (Boltz-C).
To derive relaxation rates between species, it is standard to consider a two species mixture.
In this setting \cite{francisquez2022improved,morse1963},
    \begin{subequations}
    \label{eq:BoltzC_relaxationRates}
    \begin{align}
\left.
  \frac{\partial}{\partial t}
  \left(
    \rho_i\bfu_i-\rho_j\bfu_j
  \right)
\right|_{\textnormal{(Boltz-C)}}
  &= 
    \xi_{i,j}(m_i+m_j)(\bfu_j-\bfu_i)
,
    \\
\left.
  \frac{\partial}{\partial t}
  \frac{d}{2}
  \left(
    n_iT_i-n_jT_j
  \right)
\right|_{\textnormal{(Boltz-C)}}
  &= 
    \xi_{i,j}
  \left[
    d(T_j-T_i)
    +
    \frac{m_j-m_i}{2}|\bfu_i-\bfu_j|^2
  \right]
,
    \end{align}
    \end{subequations}
where $\xi_{i,j}$ depends on the vacuum permittivity $\epsilon_0$, the Coulomb logarithm $\log( \Lambda_{i,j})$ and the species charges $q_i$ and $q_j$:
\footnote{The term $\xi_{i,j}\coloneqq\alpha_E$, where $\alpha_E$ is taken from Equation (2.5) of \cite{francisquez2022improved}.}
    \begin{equation}
    \label{eq:xiDefinition}
\xi_{i,j}
  = 
    \frac
        {2}
        {3(2\pi)^{3/2}\epsilon_0^2}
    \frac
        {|\log \Lambda_{i,j}|(q_iq_j)^2n_in_j}
        {m_im_j
          \left(
            \frac{T_i}{m_i}+\frac{T_j}{m_j}
          \right)^{3/2}
        }
.
    \end{equation}
We will match these rates with the relaxation rates of the M-LB model:
    \begin{subequations}
    \label{eq:MLB_relaxationRates}
    \begin{align}
    \label{eq:MLB_momentum}
\left.
  \frac{\partial}{\partial t}
  \left(
    \rho_i\bfu_i-\rho_j\bfu_j
  \right)
\right|_{\textnormal{M-LB}}
  &= 
    2\rho_i\lambda_{i,j}(1-\alpha_{i,j})(\bfu_j-\bfu_i)
,
    \\
    \label{eq:MLB_temperature}
\left.
  \frac{\partial}{\partial t}
  \frac{d}{2}
  \left(
    n_iT_i-n_jT_j
  \right)
\right|_{\textnormal{M-LB}}
  &= 
    2dn_i\lambda_{i,j}(1-\beta_{i,j})(T_j-T_i)
    +
    (2n_i\lambda_{i,j}\gamma_{i,j}-\delta_{i,j})|\bfu_i-\bfu_j|^2
.
    \end{align}    
    \end{subequations}
The derivation of these relaxation rates is given in \Cref{appendix:modelRelaxationRates}; specifically \eqref{eq:MLB_momentum} is verified in \Cref{lemma:momentumRelaxationRates} and \eqref{eq:MLB_temperature} is verified in \Cref{lemma:temperatureRelaxationRates}.  

Equating the right-hand sides of \eqref{eq:MLB_relaxationRates} and \eqref{eq:BoltzC_relaxationRates} gives the following matching conditions for the momentum:
    \begin{equation}
    \label{eq:momentum_matching}
2\rho_i\lambda_{i,j}(1-\alpha_{i,j}) = \xi_{i,j}(m_i + m_j)
,
    \end{equation}
and for the temperature:
    \begin{equation}
    \label{eq:temperature_matching}
2dn_i\lambda_{i,j}(1-\beta_{i,j}) 
  = d\xi_{i,j}
\qquand
2n_i\lambda_{i,j}\gamma_{i,j} - \delta_{i,j}
  = \xi_{i,j}\frac{m_j-m_i}{2}
.
    \end{equation}

    \begin{theorem}
    \label{theorem:matching}
For $1\leq i,j \leq N$, the relaxation rates in \eqref{eq:MLB_relaxationRates} match the relaxation rates in \eqref{eq:BoltzC_relaxationRates}, with the following expressions for $\alpha_{i,j}$, $\beta_{i,j}$, and $\gamma_{i,j}$:
    \begin{subequations}
    \label{eq:alphaBetaGamma_defs}
    \begin{align}
    \label{eq:alphaDef}
\alpha_{i,j}
  &= 
    1-\frac{1}{2}\frac{m_i+m_j}{m_i}\frac{\xi_{i,j}}{n_i\lambda_{i,j}},
    \\
    \label{eq:betaDef}
\beta_{i,j}
  &= 
    1-\frac{1}{2}\frac{\xi_{i,j}}{n_i\lambda_{i,j}}
  = 
    1-\frac{m_i}{m_i+m_j}(1-\alpha_{i,j}),
    \\
    \label{eq:gammaDef}
\gamma_{i,j}
  &= 
    \frac{1}{2}\frac{1}{n_i\lambda_{i,j}}
      \left[
        \xi_{i,j}\frac{m_j-m_i}{2}
        +\delta_{i,j}
      \right]
  =
    \frac{m_im_j}{m_i+m_j}(1-\alpha_{i,j})
.
    \end{align}
    \end{subequations}
    \end{theorem}
    
    \begin{proof}
The equations in \eqref{eq:alphaBetaGamma_defs} follow immediately by solving \eqref{eq:momentum_matching} and \eqref{eq:temperature_matching} for $\alpha_{i,j}$, $\beta_{i,j}$, and $\gamma_{i,j}$.
    \end{proof}
    
    \begin{remark}
The matching conditions between the coefficients in \eqref{eq:MLB_relaxationRates} and the coefficients in \eqref{eq:BoltzC_relaxationRates} are consistent with the relationship between $\alpha_{i,j}$ and $\alpha_{j,i}$ given in \Cref{prop:parameters_ji}.
Thus it is sufficient to perform the matching only for $i < j$.
More specifically, if for all $1 \leq i < j \leq N$, it holds that
    \begin{equation}
2\rho_i\lambda_{i,j}(1-\alpha_{i,j}) = \xi_{i,j}(m_i + m_j)
,
    \end{equation}
then \eqref{eq:deltaSymmetry} and the symmetry of $\xi_{i,j}$ imply that
    \begin{equation}
2\rho_j\lambda_{j,i}(1-\alpha_{j,i}) 
    = 2\rho_i\lambda_{i,j}(1-\alpha_{i,j}) 
    = \xi_{i,j}(m_i + m_j)
    = \xi_{j,i}(m_j + m_i).
    \end{equation}
That is, the matching condition holds for the momentum relaxation with $i>j$.
Similarly, \eqref{eq:betaSymmetry} and the symmetry of $\xi_{i,j}$ imply that
    \begin{equation}
d\xi_{i,j}
    = 2dn_i\lambda_{i,j}(1-\beta_{i,j}) 
    = 2dn_j\lambda_{j,i}(1-\beta_{j,i})
    = d \xi_{j,i}
,
    \end{equation}
while \eqref{eq:gammaSymmetry}, and the symmetry of $\xi_{i,j}$ and $\delta_{i,j}$ imply that
    \begin{equation}
\xi_{i,j} \frac{m_j-m_i}{2} 
    = 2n_i\lambda_{i,j}\gamma_{i,j}-\delta_{i,j}
    = -(2n_j\lambda_{j,i}\gamma_{j,i} -\delta_{j,i})
    = - \xi_{j,i} \frac{m_i-m_j}{2}
.
    \end{equation}
    \end{remark}

Based on \eqref{eq:alphaBetaGamma_defs}, it is clear that for $1 \leq i,j \leq N$, the parameters $\beta_{i,j}$ and $\gamma_{i,j}$ are completely determined by $\alpha_{i,j}$.
Moreover, we have the following

    \begin{lemma}
    \label{lemma:beta_gamma}
Suppose that for $1 \leq i,j \leq N$, $0 \leq \alpha_{i,j} \leq 1$.
Then the remaining conditions in \Cref{ass1} hold.
That is
    \begin{equation}
0 \leq \beta_{i,j} \leq 1 
\qquand
\gamma_{i,j}\geq0
,\qquad 1 \leq i,j \leq N
.
    \end{equation}
    \end{lemma}

    \begin{proof}
The proof follows immediately from \eqref{eq:betaDef} and \eqref{eq:gammaDef} and the assumed bounds on $\alpha_{i,j}$.
    \end{proof}

    \subsection{Choosing collision frequencies}
    
What remains to complete the model is to find frequencies $\lambda_{i,j}$ such that parameters $\alpha_{i,j}$ defined in \eqref{eq:alphaDef} satisfy the relation \eqref{eq:deltaSymmetry} and the condition $0 \leq \alpha_{i,j} \leq 1$.  
We propose collision frequencies of the form
    \begin{equation}
    \label{eq:collisionFrequencyDefinition}
\lambda_{i,j}
  =
    \widehat{\lambda}_{i,j}(T_i,T_j)
  = 
    \frac{\kappa_{i,j}\xi_{i,j}}{n_i}
,
    \end{equation}
where $\xi_{i,j}$ is given in \eqref{eq:xiDefinition} and $\kappa_{i,j} \geq 0$ is a tunable parameter.
Inserting \eqref{eq:collisionFrequencyDefinition} into \eqref{eq:alphaDef} gives
    \begin{equation}
    \label{eq:alpha_of_kappa}
\alpha_{i,j}
  =
    1-\frac{1}{2}\frac{m_i+m_j}{m_i}\frac{1}{\kappa_{i,j}},
    \end{equation}
which implies that $\alpha_{i,j} \leq 1$, and 
    \begin{equation}
    \label{eq:mu}
\alpha_{i,j} \geq 0
  \quiff
\kappa_{i,j} \geq \frac{m_i + m_j}{2 m_i}
  \eqqcolon \mu_{i,j}
.
    \end{equation}
Below we show how to enforce the desired conditions on $\alpha_{i,j}$ based on the choice of $\kappa_{i,j}$.

    \begin{theorem}
    \label{thm:c}
Suppose that the collision frequencies are given by  \eqref{eq:collisionFrequencyDefinition}, and $\kappa_{i,j} \geq \mu_{i,j}$ for each $1\leq i,j \leq N$.  
Then for $1\leq i,j \leq N$, $0 \leq \alpha_{i,j} \leq 1$, and for $1\leq i < j \leq N$, the pairs $(\alpha_{i,j}, \alpha_{j,i})$, $(\beta_{i,j}, \beta_{j,i})$, and $(\gamma_{i,j}, \gamma_{j,i})$ given in \eqref{eq:alphaBetaGamma_defs} satisfy \eqref{eq:modelSymmetries}. 
    \end{theorem}
    
    \begin{proof}
The bound  $0 \leq \alpha_{i,j} \leq 1$ is clear from \eqref{eq:alpha_of_kappa} and \eqref{eq:mu}.  
Moreover rearranging \eqref{eq:alpha_of_kappa} gives
    \begin{equation}
    \label{eq:alpha_of_kappa_2}
1-\alpha_{i,j}
  =
    \frac{\mu_{i,j}}{\kappa_{i,j}}
.
    \end{equation}
Using \eqref{eq:collisionFrequencyDefinition} and \eqref{eq:alpha_of_kappa_2}, the left-hand side of \eqref{eq:deltaSymmetry} is
    \begin{equation}
\rho_{i} \lambda_{i,j} (1- \alpha_{i,j}) 
  = 
    \rho_{i} \frac{\kappa_{i,j}\xi_{i,j}}{n_i} \frac{\mu_{i,j}}{\kappa_{i,j}}
  = 
    \xi_{i,j}\frac{m_i + m_j}{2}
,
    \end{equation}
and the right-hand side is
    \begin{equation}
\rho_{j} \lambda_{j,i} (1- \alpha_{j,i}) 
  = 
    \rho_{j} \frac{\kappa_{j,i}\xi_{j,i}}{n_j} \frac{\mu_{j,i}}{\kappa_{j,i}}
  =
    \xi_{j,i}\frac{m_j + m_i}{2}
.
    \end{equation}
Since $\xi_{i,j}=\xi_{j,i}$, the symmetry in \eqref{eq:deltaSymmetry} holds.

Next, using \eqref{eq:betaDef}, \eqref{eq:collisionFrequencyDefinition}, and \eqref{eq:alpha_of_kappa_2},
    \begin{equation}
1-\beta_{i,j}
  = 
    \frac{1}{2\kappa_{i,j}}
,    
    \end{equation}
so that the left-hand side of \eqref{eq:betaSymmetry} is
    \begin{equation}
n_i\lambda_{i,j}(1-\beta_{i,j})
  =
    n_i\frac{\kappa_{i,j}\xi_{i,j}}{n_i}\frac{1}{2\kappa_{i,j}}
  =
    \frac{\xi_{i,j}}{2}
,
    \end{equation}
and the right-hand side is
    \begin{equation}
n_i\lambda_{j,i}(1-\beta_{j,i})
  =
    n_j\frac{\kappa_{j,i}\xi_{j,i}}{n_j}\frac{1}{2\kappa_{j,i}}
  =
    \frac{\xi_{j,i}}{2}
.
    \end{equation}
Since $\xi_{i,j}=\xi_{j,i}$, the symmetry in \eqref{eq:betaSymmetry} holds.

Finally, using \eqref{eq:gammaDef}, \eqref{eq:collisionFrequencyDefinition}, and \eqref{eq:alpha_of_kappa_2},
    \begin{equation}
\gamma_{i,j}
  =
    \frac{m_j}{2\kappa_{i,j}}
,
    \end{equation}
so that the left-hand side of \eqref{eq:gammaSymmetry} is
    \begin{equation}
\delta_{i,j}-2n_i\lambda_{i,j}\gamma_{i,j}
  =
    m_in_i\frac{\kappa_{i,j}\xi_{i,j}}{n_i}\frac{\mu_{i,j}}{\kappa_{i,j}} 
    -
    2n_i\frac{\kappa_{i,j}\xi_{i,j}}{n_i}\frac{m_j}{2\kappa_{i,j}}
  =
    \xi_{i,j}\frac{m_i-m_j}{2}
,
    \end{equation}
and the right-hand side is
    \begin{equation}
2n_j\lambda_{j,i}\gamma_{j,i} - \delta_{j,i}
  =
    2n_j\frac{\kappa_{j,i}\xi_{j,i}}{n_j}\frac{m_i}{2\kappa_{j,i}}
    -
    m_jn_j\frac{\kappa_{j,i}\xi_{j,i}}{n_j}\frac{\mu_{j,i}}{\kappa_{j,i}}
  =
    \xi_{j,i}\frac{m_i-m_j}{2}
.
    \end{equation}
Since $\xi_{i,j}=\xi_{j,i}$, the symmetry in \eqref{eq:gammaSymmetry} holds.
    \end{proof}

    \begin{corollary}
For every $1 \leq i, j \leq N$, let $\alpha_{i,j}$, $\beta_{i,j}$, and $\gamma_{i,j}$ be given by \eqref{eq:alphaBetaGamma_defs}, and let $\lambda_{i,j}$ be given by \eqref{eq:collisionFrequencyDefinition}, where $\kappa_{i,j}$ satisfies \eqref{eq:mu}.  
Then (i) $\alpha_{i,j}$, $\beta_{i,j}$, and $\gamma_{i,j}$ satisfy \Cref{ass1}.
Furthermore, any solution $\bff = [f_i ,\cdots, f_N]^\top$ of the M-LB model \eqref{eq:MLB}-\eqref{eq:MLB_operator} (ii) satisfies the pairwise-wise conservation laws (\Cref{ass:conservation}), (iii) satisfies an entropy dissipation law and an $\cH$-Theorem  (\Cref{theorem:HTheorem}), and (iv) matches the pairwise momentum and temperature relaxation rates derived from the Boltzmann collision operator with Coulomb potential (\Cref{theorem:matching}).  
    \end{corollary}

    \begin{proof}
Claim (i) is a restatement of \Cref{lemma:beta_gamma}, along with the bound in \eqref{eq:mu}.
Claim (ii) follows from \Cref{thm:c} and the equivalence of \Cref{ass:conservation} and \eqref{eq:ji_parameter_defs} (see \Cref{prop:parameters_ji}).
Claims (iii) and (iv) are summaries of the theorems referenced in the statement of the corollary.
    \end{proof}
    
    \section{Numerics}
    \label{section:numerics}
    
For simulation purposes, we propose an implicit Euler update of \eqref{eq:MLB} which takes the form (cf. \eqref{eq:MLB_operator})
    \begin{equation}
\frac{f_i^{\mfn+1}-f_i^\mfn}{\dt}
  = 
    \sum_j\lambda_{i,j}^{\mfn+1}\theta_{i,j}^{\mfn+1}\cL_{i,j}^{\mfn+1},
\;\;\textnormal{where}\;\;
\cL_{i,j}^{\mfn+1}
  =
    \nabla_{\bfv}\cdot
    \left(
      M_{i,j}^{\mfn+1}\nabla_{\bfv}
      \left(
        \frac{f_i^{\mfn+1}}{M_{i,j}^{\mfn+1}}
      \right)
    \right)
.   
    \label{eq:MLB_BEStep}
    \end{equation}
Because $\lambda_{i,j}^{\mfn+1}$, $\theta_{i,j}^{\mfn+1}$, and the mixture Maxwellians, $M_{i,j}^{\mfn+1}$ in $\cL^{\mfn+1}_{i,j}$ are defined in terms of the moment quantities $\rho_i^{\mfn+1}$, $\bfu_i^{\mfn+1}$, and $T_i^{\mfn+1}$, it is useful to first derive and solve a system of equations for these moments.
What remains afterward is the inversion of a linear tridiagonal system, the coefficients of which are expressed in terms of the updated moments.

    \subsection{Implicit moment update}
    \label{section:momentODEsystem}
    
To derive a set of moment equations for $\rho^\mfn_i$, $\bfu^\mfn_i$, and $T_i^\mfn$, we multiply \eqref{eq:MLB_BEStep} by $m_i(1,\bfv,|\bfv|^2)^\top$ and integrate over $\bfv \in \bbR^d$.
Then, applying \Cref{lemma:momentProperties} and the definitions of $\bfu_{i,j}$ and $T_{i,j}$ in \eqref{eq:u_ij_and_T_ij_defs}, gives the following Backward Euler update 
    \begin{subequations}
    \label{eq:BEUpdate}
    \begin{align}
    \label{eq:densityUpdate}
\frac{\rho_i^{\mfn+1}-\rho_i^\mfn}{\dt}
  &= 0
,
    \\
    \label{eq:momentumUpdate}
\rho_i
\frac{\bfu_i^{\mfn+1}-\bfu_i^\mfn}{\dt}
  &=
    \sum_j
    \rho_i\lambda_{i,j}^{\mfn+1}(1-\alpha_{i,j})
    (\bfu_j^{\mfn+1}-\bfu_i^{\mfn+1})
,
    \\
    \label{eq:energyUpdate}
    \begin{split}
\rho_i
\frac{s_i^{\mfn+1}-s_i^\mfn}{\dt}
  +
dn_i
\frac{T_i^{\mfn+1}-T_i^\mfn}{\dt}
  &=
    2d\sum_j
    n_i\lambda_{i,j}^{\mfn+1}(1-\beta_{i,j})
    (T_j^{\mfn+1}-T_i^{\mfn+1})
    \\
  + 2 \sum_jn_i\lambda_{i,j}^{\mfn+1}(1-\beta_{i,j}&)
     \big[
        m_j(
        s_j^{\mfn+1}-\bfu_i^{\mfn+1}\cdot\bfu_j^{\mfn+1})
        -
        m_i(s_i^{\mfn+1}-\bfu_i^{\mfn+1}\cdot\bfu_j^{\mfn+1})
      \big]
,
    \end{split}
    \end{align}
    \end{subequations}
where $s_i^\mfn=|\bfu_i^\mfn|^2$, and we drop the superscripts on $n_i$ and $\rho_i$ due to \eqref{eq:densityUpdate}.
The system given here is consistent with the backward Euler step of the ODE system \eqref{eq:momentODEsystem}, when the matching conditions in \eqref{eq:alphaBetaGamma_defs} hold.
Summing \eqref{eq:momentumUpdate} and \eqref{eq:energyUpdate} over $i\in\{1,\cdots,N\}$ and using the symmetries \eqref{eq:deltaSymmetry} and \eqref{eq:betaSymmetry}
implies that the discrete versions of the conservation of total momentum and total energy hold:
    \begin{equation}
\sum_i\rho_i\bfu_i^{\mfn+1} = \sum_i\rho_i\bfu_i^\mfn
\qquand
\sum_iE_i^{\mfn+1} = \sum_iE_i^\mfn
.
    \end{equation}

The form of the implicit update in \eqref{eq:BEUpdate} is almost identical to the moment system in \cite{implicitPaper} derived for a multi-species BGK model.
The only difference is the factor of two in the energy equation and the appearance of $\bfu_i^{\mfn+1}\cdot\bfu_j^{\mfn+1}$ instead of $|\bfu_{i,j}^{\mfn+1}|^2$.
For this reason, we adapt the nonlinear Gauss-Seidel type (GST) iterative method of \cite{implicitPaper} to solve \eqref{eq:BEUpdate}.
Given iteration index $\ell$, the iterative update for the momentum and energy equations within each time step takes the form    
    \begin{subequations}
    \label{eq:GST}
    \begin{align}
&\rho_i
\frac{\bfu_i^{\mfn+1,\ell+1}-\bfu_i^\mfn}{\dt}
  =
    \sum_j
    \rho_i\lambda_{i,j}^{\mfn+1,\ell}(1-\alpha_{i,j})
    (\bfu_j^{\mfn+1,\ell+1}-\bfu_i^{\mfn+1,\ell+1})
    \\
&\rho_i
\frac{s_i^{\mfn+1,\ell+1}\!-s_i^\mfn}{\dt}
  +
dn_i
\frac{T_i^{\mfn+1,\ell+1}\!-T_i^\mfn}{\dt}
  =
    2d\sum_j
    n_i\lambda_{i,j}^{\mfn+1,\ell}(1-\beta_{i,j})
    (T_j^{\mfn+1,\ell+1}\!-T_i^{\mfn+1,\ell+1})
    \\
&  + 2 \sum_jn_i\lambda_{i,j}^{\mfn+1,\ell}(1-\beta_{i,j})
     \bigg[
        m_j(s_j^{\mfn+1,\ell+1}-\bfu_i^{\mfn+1,\ell+1}\cdot\bfu_j^{\mfn+1,\ell+1})
    -
        m_i(s_i^{\mfn+1,\ell+1}-\bfu_i^{\mfn+1,\ell+1}\cdot\bfu_j^{\mfn+1,\ell+1})
      \bigg]
,    
    \nonumber
    \end{align}
    \end{subequations}
where
    \begin{equation}
\lambda_{i,j}^{\mfn+1,\ell}
  = 
    \widehat{\lambda}_{i,j}
      (T_i^{\mfn+1,\ell}, T_j^{\mfn+1,\ell})
,  
    \end{equation}
and $\widehat{\lambda}_{i,j}$ is given by \eqref{eq:collisionFrequencyDefinition} and \eqref{eq:xiDefinition}.
Once $\bfu^{\mfn+1}_i$ and $T^{\mfn+1}_i$ have been evaluated using the GST method \eqref{eq:GST}, they can be used to update the mixture terms $\bfu^{\mfn+1}_{i,j}$ and $T^{\mfn+1}_{i,j}$ via \eqref{eq:u_ij_and_T_ij_defs}, as well as the Maxwellian terms $M_{i,j}^{\mfn+1}$ via \eqref{eq:mixtureMaxwellian_and_theta_ij_def}.

    \subsection{Collision term update}

After $\bfu_i^{\mfn+1}$ and $T_i^{\mfn+1}$ are determined, \eqref{eq:MLB_BEStep} is a linear equation.
For simplicity, we focus on the case $d=1$, but the discretization below can be extended to a general dimension $d$.
We truncate the velocity space at some sufficiently large value, $V=[-v_{\max},v_{\max}]$, and consider a cell centered velocity grid with $N_v$ grid cells of uniform size
$h_v = 2v_{\max}/N_v$, with $v_{k} = -v_{\max} + (k-1/2)h_v$
defined for $k$ at integer and half integer values, so that $v_{1/2}=-v_{\max}$ and $v_{N_v+1/2}=v_{\max}$.

Approximating velocity derivatives in $\cL_{i,j}$ by central differences gives, for $k\in\{1,\cdots,N_v\}$,
    \begin{equation}
    \label{eq:discreteCollisionOperator}
    \begin{split}
\left[
  \cL_{i,j}
\right]_{k}
  &=
    \left[
      \diff_{v}
      \left(
        M_{i,j}\diff_{v}
        \left(
          \frac{f_i}{M_{i,j}}
        \right)
      \right)
    \right]_{k}
    \\
  &\approx
    \frac{[M_{i,j}]_{k+\frac{1}{2}}[\diff_{v}(f_i/M_{i,j})]_{k+\frac{1}{2}} - [M_{i,j}]_{k-\frac{1}{2}}[\diff_{v}(f_i/M_{i,j})]_{k-\frac{1}{2}}}{\Delta v}
    \\
  &\approx
    \frac{[M_{i,j}]_{k+\frac{1}{2}}
    \left[
      \left[
        \frac{f_i}{M_{i,j}}
      \right]_{k+1}
      -
      \left[
        \frac{f_i}{M_{i,j}}
      \right]_{k}
    \right]
    -
    [M_{i,j}]_{k-\frac{1}{2}}
    \left[
      \left[
        \frac{f_i}{M_{i,j}}
      \right]_{k}
      -
      \left[
        \frac{f_i}{M_{i,j}}
      \right]_{k-1}
    \right]}{\Delta v^2}
    \\
  &=
    \frac{
      a_{k-1}[f_i]_{k-1}
      -
      b_{k}[f_i]_{k}
      +
      c_{k+1}[f_i]_{k+1}}{ \Delta v^2}
,
    \end{split}
    \end{equation}
where
    \begin{equation}
a_{k-1}
  =
    \frac{[M_{i,j}]_{k-\frac{1}{2}}}{[M_{i,j}]_{k-1}}
,\qquad
b_k 
  = 
    \frac{[M_{i,j}]_{k-\frac{1}{2}}}{[M_{i,j}]_{k}}
    +
    \frac{[M_{i,j}]_{k+\frac{1}{2}}}{[M_{i,j}]_{k}}
,\qquad
c_{k+1}
  =
    \frac{[M_{i,j}]_{k+\frac{1}{2}}}{[M_{i,j}]_{k+1}}
,
    \end{equation}
and the dependence of $a_{k}$, $b_{k}$, and $c_{k}$ on $i$ and $j$ has been suppressed to simplify notation.
If we set $[M_{i,j}]_{k+\frac{1}{2}}=\frac{1}{2}([M_{i,j}]_{k}+[M_{i,j}]_{k+1})$, then
    \begin{equation}
a_{k-1}
  = \frac{[M_{i,j}]_{k-1}+[M_{i,j}]_{k}}{2[M_{i,j}]_{k-1}}
,\quad
b_{k}
  =
    \frac{[M_{i,j}]_{k-1} + 2[M_{i,j}]_{k} + [M_{i,j}]_{k+1}}{2[M_{i,j}]_{k}}
,\quad
c_{k+1}
  = \frac{[M_{i,j}]_{k}+[M_{i,j}]_{k+1}}{2[M_{i,j}]_{k+1}}
.
    \end{equation}

To obtain no-flux boundary conditions (for velocity space), we set $[M_{i,j}]_{\frac{1}{2}}=0=[M_{i,j}]_{N_v+\frac{1}{2}}$, which implies that
    \begin{equation}
a_{0}
  = 
    \frac{[M_{i,j}]_{\frac{1}{2}}}{[M_{i,j}]_{0}} = 0
,\quad\!
b_{1}
  =
    \frac{[M_{i,j}]_{1+\frac{1}{2}}}{[M_{i,j}]_{1}}
,\quad\!
b_{N_v}  
  =
    \frac{[M_{i,j}]_{N_v-\frac{1}{2}}}{[M_{i,j}]_{N_v}}
,\quad\!
c_{N_v+1}
  =
    \frac{[M_{i,j}]_{N_v+\frac{1}{2}}}{[M_{i,j}]_{N_v+1}} = 0
.
    \end{equation}
Rearranging \eqref{eq:MLB_BEStep} and using \eqref{eq:discreteCollisionOperator} gives
    \begin{equation}
[f_i^{\mfn+1}]_{k}
  \approx
    [f_i^\mfn]_{k}
    +
    \frac{\dt}{\Delta v^2}\sum_j
    \lambda_{i,j}^{\mfn+1}\theta_{i,j}^{\mfn+1}
    \left(
      a_{k-1}^{\mfn+1}[f_i]_{k-1}^{\mfn+1}
      -
      b_{k}^{\mfn+1}[f_i]_{k}^{\mfn+1}
      +
      c_{k+1}^{\mfn+1}[f_i]_{k+1}^{\mfn+1}
    \right)
,
    \end{equation}
or in vectorized form,
    \begin{equation}
    \label{eq:implicitUpdateMatrixEq}
\left(
  I+\frac{\dt}{\Delta v^2}\sum_j\lambda_{i,j}^{\mfn+1}\theta_{i,j}^{\mfn+1}G_{i,j}^{\mfn+1}
\right)
\vec{\bff}^{\mfn+1}_i
  =
    \vec{\bff}^\mfn_i
,
    \end{equation}
where
    \begin{equation}
\vec{\bff}^{\mfn+1}_i
  =
    \bigg([f_i^{\mfn+1}]_{k=1},\cdots,[f_i^{\mfn+1}]_{k=N_v}\bigg)^\top
,
    \end{equation}
and
    \begin{equation}
G_{i,j}^{\mfn+1}
  = 
\begin{pmatrix}
  b_{1}^{\mfn+1}
  & -c_{2}^{\mfn+1}
  & 0 & \cdots & & \cdots & 0
    \\
  -a_{1}^{\mfn+1}
  & b_{2}^{\mfn+1}
  & -c_{3}^{\mfn+1} & 0 & \cdots & \cdots & 0
    \\
  0 & -a_{2}^{\mfn+1}
  & b_{3}^{\mfn+1}
  & -c_{4}^{\mfn+1} & 0 & \cdots & 0
    \\
  \vdots &  & \ddots & \ddots & \ddots & & \vdots
    \\
    \\
    \\
  0 & \cdots & & 0 & -a_{N_v-2}^{\mfn+1}
  & b_{N_v-1}^{\mfn+1}
  & -c_{N_v}^{\mfn+1}
    \\
  0 & \cdots & & & 0 & -a_{N_v-1}^{\mfn+1}
  & b_{N_v}^{\mfn+1}
    \end{pmatrix}
.
    \end{equation}
The values of the Maxwellians $M_{i,j}^{\mfn+1}$ and hence coefficients of $G_{i,j}^{\mfn+1}$ are explicitly computed with the moment values determined by the GST algorithm in \Cref{section:momentODEsystem}.
After this, the implicit update of the nonlinear collision operator in \eqref{eq:MLB_BEStep} requires only the linear inversion of \eqref{eq:implicitUpdateMatrixEq}.

    \section{Numerical examples}
    \label{section:numericalExamples}
    
To provide some numerical examples, we appeal to \cite{particleMethod_2023}, which considers a multi-species LFP collision operator.
We use simple constants rather than physically correct values for the physical parameters:
    \begin{equation}
n_1 = n_2 = 1
  ,\qquad
q_i = q_j =1
  ,\qquad
|\log\Lambda_{i,j}|=1
  ,\qquad
\epsilon_0=1
.
    \end{equation}
For both tests, the initial distributions are set to Maxwellian initial conditions with
    \begin{align}
f_1(0,v)
  &=
    M_{1\,,\,0.5\,,\,0.25}^{(1)}
\qquand
f_2(0,v)
  =
    M_{1\,,\,-0.25\,,\,0.125}^{(1)}
,
    \end{align}
where $M_{n,\bfu,\theta}^{(d)}$ is the $d$-dimensional Maxwellian defined in \eqref{eq:maxwellianDefinition}.
To ensure that $\alpha_{i,j}$, $\beta_{i,j}$, and $\gamma_{i,j}$ satisfy the desired bounds, we enforce \eqref{eq:mu}, i.e. $\kappa_{i,j}\geq\frac{m_i+m_j}{2m_i}\eqqcolon\mu_{i,j}$ for all $i,j$.

The computational domain for the velocity space is truncated to $[-4,4]$ (as in \cite{particleMethod_2023}), and is discretized using a cell-centered grid over $N_v=80$ uniform cells.
The results below were generated using $\dt = 0.2$; smaller time steps gave qualitatively similar results.

For the first test, we set $m_1=m_2=1$.
Thus $\mu_{i,j} \leq 1$, and the simulation is run with $\kappa_{i,j} = 2$ for all $i,j$.
The results are given in \Cref{figure:mRatio1}.
For the second test, we let $m_1=2$ and $m_2=1$, and set $\kappa_{i,j}=2\max\{\mu_{i,j}\}=2\mu_{2,1} = 3$.
The results are given in \Cref{figure:mRatio2}.
In both examples, the long time behavior of the moments is as expected.
That is, both species relax toward the steady-state bulk velocity and temperature given by  $\bfu^\infty$ and $T^\infty$ in \eqref{eq:u_inf_and_T_inf_defs}.

\begin{figure}[ht]
    \centering
    \begin{subfigure}[t]{0.48\textwidth}
        \centering
        \includegraphics[width=\textwidth]{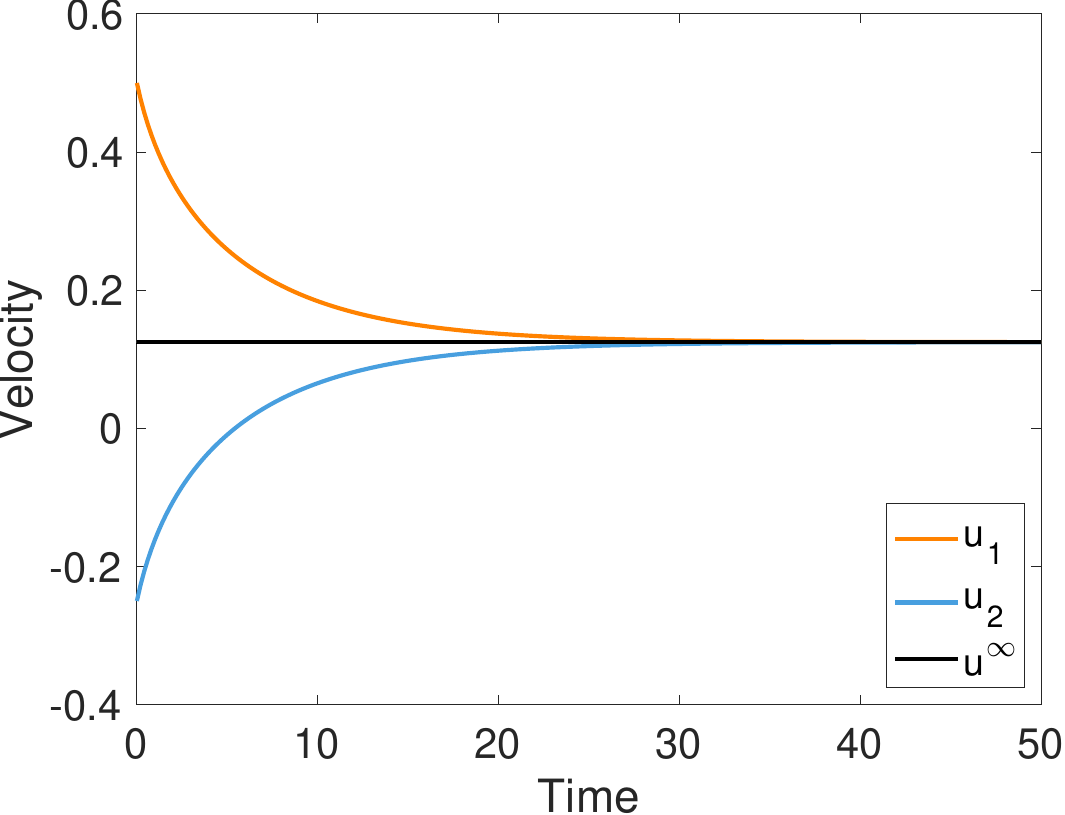}
        \caption{Species velocities relax to steady state, $\bfu^\infty$.}
    \end{subfigure}
    \hfill
    \begin{subfigure}[t]{0.48\textwidth}  
        \centering 
        \includegraphics[width=\textwidth]{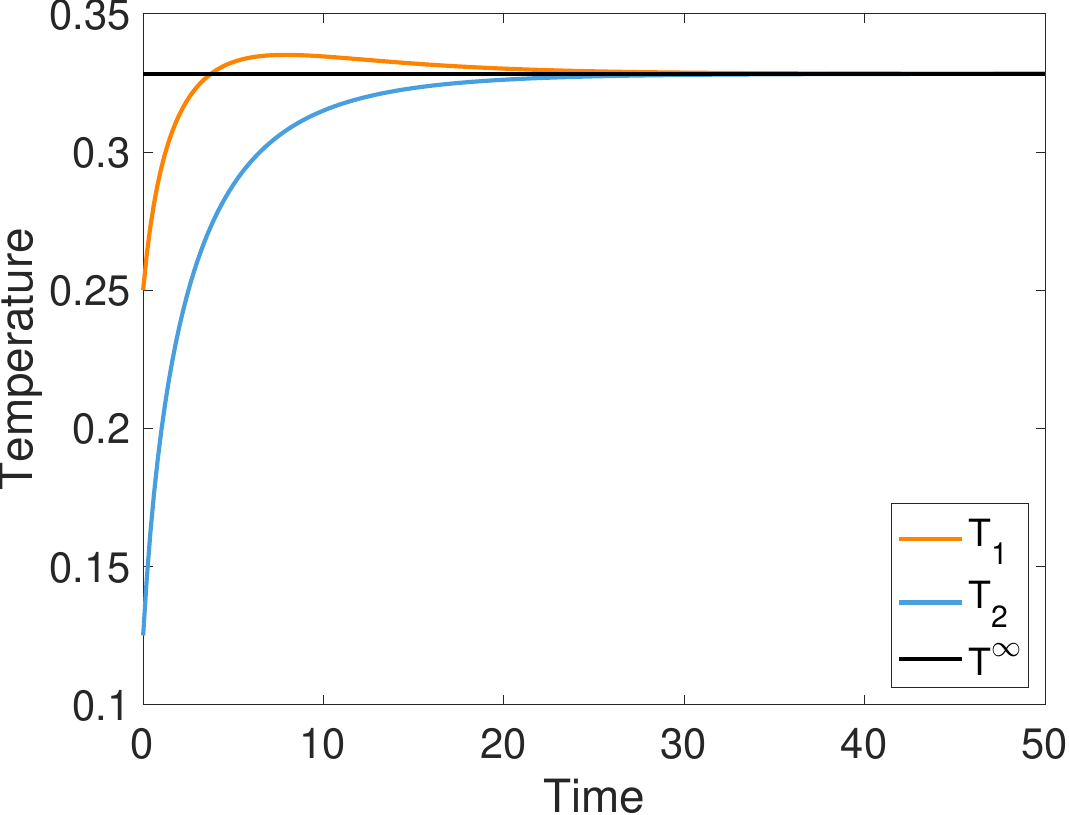}
        \caption{Species temperatures relax to steady state, $T^\infty$.}
    \end{subfigure}
    \caption{Case: Equal Masses: The velocities and temperatures relax to steady state values.}
    \label{figure:mRatio1}
\end{figure}

\begin{figure}[ht]
    \centering
    \begin{subfigure}[t]{0.48\textwidth}
        \centering
        \includegraphics[width=\textwidth]{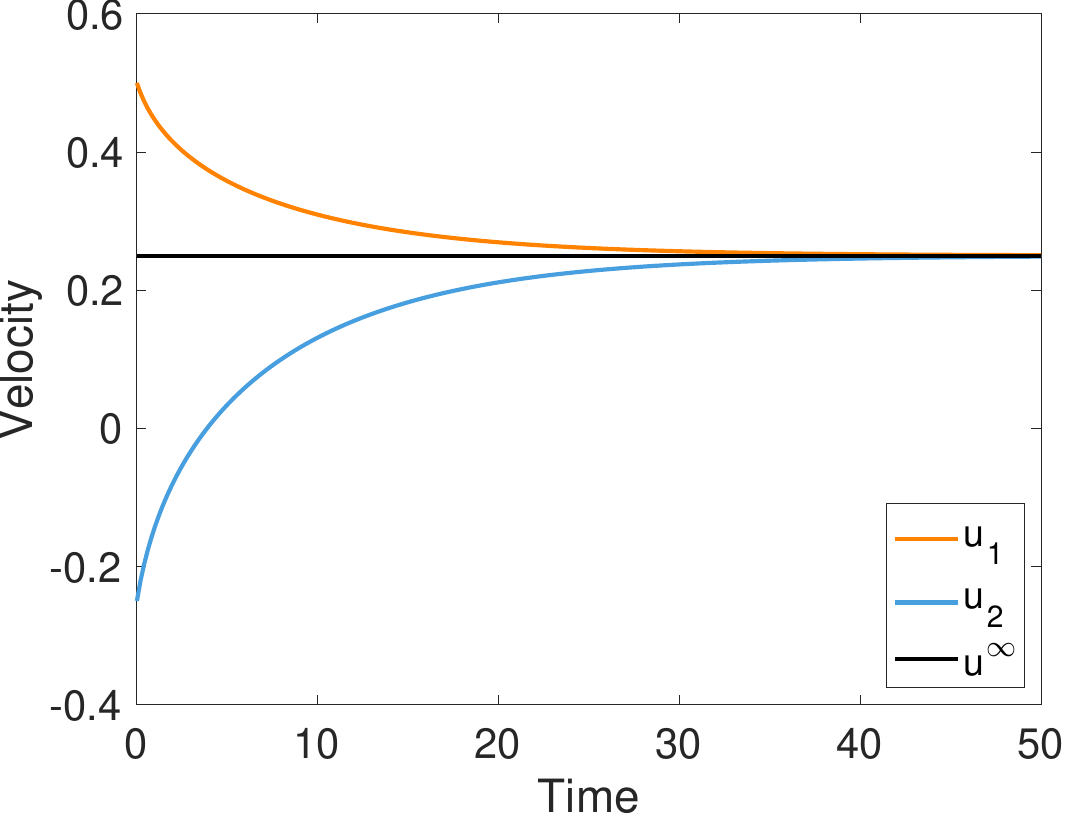}
        \caption{Species velocities relax to steady state, $\bfu^\infty$.}
    \end{subfigure}
    \hfill
    \begin{subfigure}[t]{0.48\textwidth}  
        \centering 
        \includegraphics[width=\textwidth]{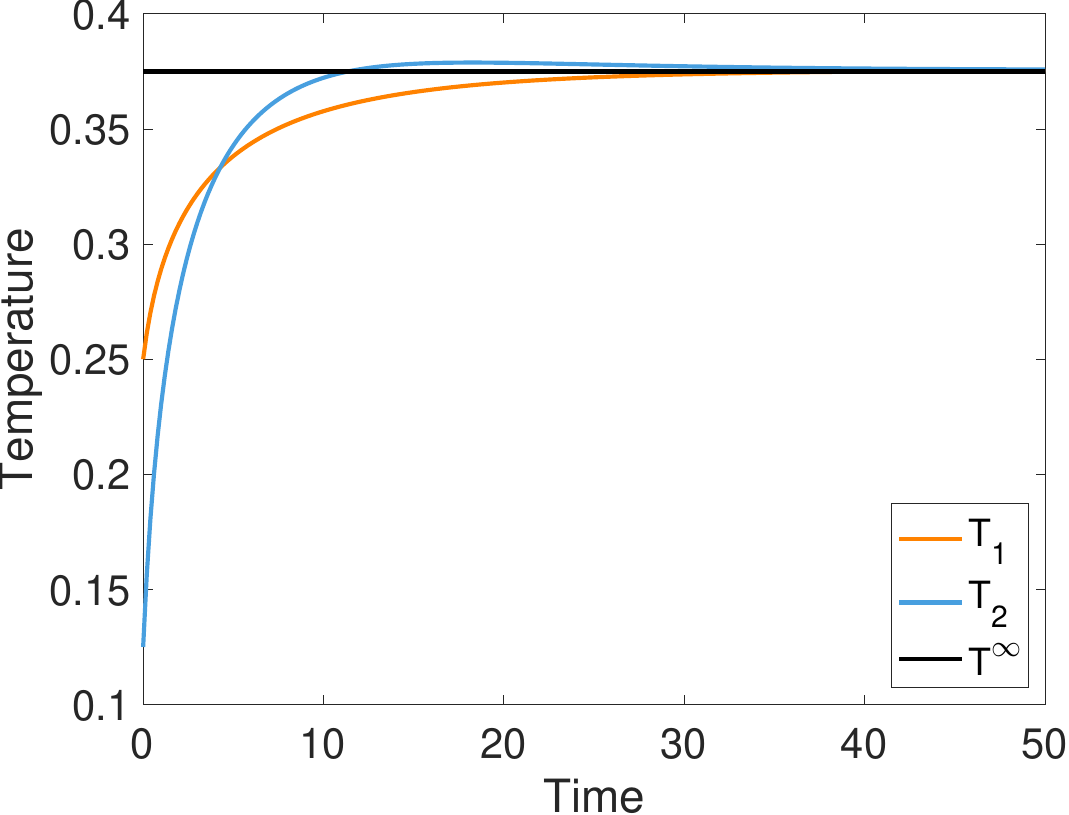}
        \caption{Species temperatures relax to steady state, $T^\infty$.}
    \end{subfigure}
    \caption{Case: Mass Ratio$= 2$. The velocities and temperatures relax to steady state values.}
    \label{figure:mRatio2}
\end{figure}
    
    \section{Conclusions}

In this paper, we have derived a multi-species Lenard-Bernstein (M-LB) model that satisfies (i) conservation of species mass and global conservation of momentum and energy; (ii) dissipates entropy and satisfies an  H-Theorem; and (iii) matches the pairwise momentum and temperature relaxation rates of the Boltzmann collision operator with a Coulomb potential.
We have also demonstrated a simple numerical algorithm for implementing the (M-LB) model with implicit time stepping.
In future work, we hope to embed this algorithm into a full kinetic simulation capability, such as \cite{endeve2022conservative} or \cite{schnake2024sparse}.

    \appendix

    \section{Comparison of models}

During the finalization of this manuscript, the authors became aware of very recent and independent work on the M-LB model \cite{pirner2024consistent}.
This work does not consider the problem of matching relaxation rates or perform any numerical experiments, but it does establish conditions on the coefficients of the mixture temperature and velocity that are sufficient for the existence of an $\cH$-Theorem.
Below we show that the conditions made in \cite{pirner2024consistent} are sufficient to imply the conditions made in \Cref{ass1}, which we believe provide a more transparent physical interpretation, and allow us to establish the entropy dissipation law and $\cH$-Theorem under slightly weaker conditions.
We provide more details below, using the notation of \cite{pirner2024consistent} when possible, but relying on the notation used here when there is a conflict.  
For clarity we consider, as in \cite{pirner2024consistent}, that $i \in \{1,2\}$, although this is not a necessary restriction for either work.

Let
    \begin{equation}
c_{1,2}
    = \frac{\lambda_{1,2}}{n_2}
,\qquad
c_{2,1}
    = \frac{\lambda_{2,1}}{n_1}
,\qquand
\varepsilon
    = \frac{c_{1,2}}{c_{2,1}} 
    = \frac{n_1}{n_2}\frac{\lambda_{1,2}}{\lambda_{1,2}}
    \geq 0
. 
    \end{equation}
In terms of $\varepsilon$, the conditions in \eqref{eq:ji_parameter_defs} become
    \begin{subequations}
    \label{eq:21_parameter_defs}
    \begin{align}
    \label{eq:alpha_21_def}
\alpha_{2,1}
  &=
    1-\frac{\rho_1\lambda_{1,2}}{\rho_2\lambda_{2,1}}(1-\alpha_{2,1}),
    \\
    \label{eq:beta_21_def}
\beta_{2,1}
  &=
    1-\frac{n_1\lambda_{1,2}}{n_2\lambda_{2,1}}(1-\beta_{1,2}),
    \\
     \label{eq:gamma_21_def}
 \gamma_{2,1}
   &=
      m_1 \varepsilon (1-\alpha_{1,2}) - \varepsilon \gamma_{1,2}
.
    \end{align}
    \end{subequations}
    
    \begin{ass}
    \label{ass:pirner}
For $i=1,2$, the coefficients $\alpha_{i,j}$, $\beta_{i,j}$, and $\gamma_{i,j}$ satisfy the following conditions:
    \begin{subequations}
    \begin{align}
    \label{eq:pirner_eps1}
&\varepsilon  \leq 1
,
    \\
    \label{eq:pirner_eps2}
&\varepsilon  \leq \frac{m_2}{m_1}
, 
    \\
    \label{eq:pirner_gamma}
&0  \leq \gamma_{1,2} \leq m_1(1 - \alpha_{1,2})
,
    \\
    \label{eq:pirner_alpha}
&\frac{\varepsilon}{1+\varepsilon}
  \leq
    \alpha_{1,2}
  \leq
    1
,
    \\
    \label{eq:pirner_beta}
&\frac{\varepsilon}{1+\varepsilon}
  \leq
    \beta_{1,2} 
  \leq
    1
.
    \end{align}
    \end{subequations}
    \end{ass}
    
    \begin{remark}
We list the conditions above in the order they appear in \cite{pirner2024consistent}.
The conditions in \eqref{eq:pirner_eps1} and \eqref{eq:pirner_eps2} are given in \cite[Eq. 13]{pirner2024consistent}.
The condition in \eqref{eq:pirner_gamma} is given in \cite[Eq. 20]{pirner2024consistent}.
The conditions in \eqref{eq:pirner_alpha} and \eqref{eq:pirner_beta} are given in \cite[Eq. 23]{pirner2024consistent}.
A third condition for entropy dissipation and an H-Theorem is given in \cite[Eq. 23]{pirner2024consistent}, but we do not need it here.
    \end{remark}

    \begin{theorem}
Suppose the pairwise conservation laws in \Cref{ass:conservation} hold.
Then for $i=1,2$, the conditions in \Cref{ass:pirner} imply the conditions in \Cref{ass1}.
    \end{theorem}
    
    \begin{proof}
The conditions in \eqref{eq:pirner_alpha} and \eqref{eq:pirner_beta} clearly imply that $0 \leq \alpha_{1,2}, \beta_{1,2} \leq 1$.
The condition in \eqref{eq:pirner_eps1} is equivalent to
    \begin{equation}
    \label{eq:MP13_1}
\frac{n_1}{n_2}\frac{\lambda_{1,2}}{\lambda_{2,1}}
  \leq
    1
.
    \end{equation}
According to \Cref{prop:parameters_ji}, the pairwise conservation of momentum and energy in \Cref{ass:conservation} are equivalent to \eqref {eq:21_parameter_defs}.
Thus according to \eqref{eq:beta_21_def}, \eqref{eq:MP13_1}, and the fact that $\beta_{1,2} \in  [0,1]$, it holds that $\beta_{2,1} \in [0,1]$.
Similarly, the condition in \eqref{eq:pirner_eps1} is equivalent to
    \begin{equation}
    \label{eq:MP13_2}
\frac{\rho_1}{\rho_2}
\frac{\lambda_{1,2}}{\lambda_{2,1}}
  \leq
    1
. 
    \end{equation}
Acccording to \eqref{eq:alpha_21_def}, \eqref{eq:MP13_2}, and the fact that $\alpha_{1,2} \in [0,1]$, it holds that $\alpha_{2,1} \in [0,1]$.
Finally, \eqref{eq:pirner_gamma} states directly that $\gamma_{1,2} \geq 0$.
Moreoever, when combined with \eqref{eq:gamma_21_def}, the upper bound in \eqref{eq:pirner_gamma} implies that $\gamma_{2,1} \geq 0$.
    \end{proof}

    \section{Model relaxation rates}
    \label{appendix:modelRelaxationRates}

In this appendix, we prove \eqref{eq:MLB_relaxationRates}, in \Cref{lemma:momentumRelaxationRates} and \Cref{lemma:temperatureRelaxationRates}.
    
    \begin{lemma}
    \label{lemma:momentumRelaxationRates}
Momentum relaxation rates between species $i$ and $j$ are given by
    \begin{equation}
    \label{eq:momentumDiffRelax}
\left.
  \frac{\partial}{\partial t}
  \left(
    \rho_i\bfu_i-\rho_j\bfu_j
  \right)
\right|_{\textnormal{M-LB}}
  = 
    2\rho_i\lambda_{i,j}(1-\alpha_{i,j})(\bfu_j-\bfu_i)
.
    \end{equation}
    \end{lemma}

    \begin{proof}
Integrate \eqref{eq:MLB} against $m_i\bfv$ and use \eqref{eq:momentProperty1} to obtain the following system of equations that describes the time evolution of the momentum moment system:
    \begin{align}
\intd \frac{\partial f_i}{\partial t}m_i\bfv\dv
  &=
    m_i\sum_{j}\intd \cC_{i,j}(f_i,f_j)\bfv\dv
  =
    \sum_jm_i\lambda_{i,j}\theta_{i,j}\intd\bfv\cL(f_i,M_{i,j})\dv
    \nonumber
    \\
    \label{eq:momentumRelaxation}
\iff 
\frac{\partial}{\partial t}(\rho_i\bfu_i)
  &= \sum_j\rho_i\lambda_{i,j}(\bfu_{i,j}-\bfu_i)
.
    \end{align}
To derive the relaxation rates between species $i$ and $j$, neglect the species in $\{1,\cdots,N\}\setminus\{i,j\}$, take the difference of relaxation rates in \eqref{eq:momentumRelaxation}, and use \eqref{eq:u_ij_def} and \eqref{eq:deltaSymmetry}
to obtain \eqref{eq:momentumDiffRelax}.
    \end{proof}
    
    \begin{lemma}
    \label{lemma:temperatureRelaxationRates}
Thermal relaxation rates between species $i$ and $j$ are given by
    \begin{equation}
    \label{eq:temperatureDiffRelax_MLB}
\left.
  \frac{\partial}{\partial t}\frac{d}{2}
  \left(
    n_iT_i-n_jT_j
  \right)
\right|_{\textnormal{M-LB}}
  =
    2dn_i\lambda_{i,j}(1-\beta_{i,j})
    (T_j-T_i)
    +
    (2n_i\lambda_{i,j}\gamma_{i,j} - \delta_{i,j})
    |\bfu_i-\bfu_j|^2
.
    \end{equation}
    \end{lemma}
    
    \begin{proof}
Integrating \eqref{eq:MLB} against $m_i|\bfv|^2/2$, and recalling the definition of the energy $E_i$, in \eqref{eq:energyDefinition}, gives:
    \begin{align}
    \label{eq:energyRelaxation}
\intd\frac{\partial f_i}{\partial t} m_i\frac{|\bfv|^2}{2}\dv
  &=
    \sum_jm_i\lambda_{i,j}\theta_{i,j}\intd \frac{|\bfv|^2}{2}\cL(f_i,M_{i,j})\dv
    \\
\iff
\frac{\partial}{\partial t}
  \left(
    \frac{1}{2}\rho_i|\bfu_i|^2 +\frac{d}{2}n_iT_i
  \right)
  &= 
    \sum_j\rho_i\lambda_{i,j}\bfu_i \cdot(\bfu_{i,j}-\bfu_i)
    +
    d\sum_jn_i\lambda_{i,j}(T_{i,j}-T_i)
.
    \nonumber
    \end{align}
Using \eqref{eq:momentumRelaxation}, the first term is
    \begin{equation}
    \label{eq:d_tKineticEnergy}
\frac{\partial}{\partial t}
  \left(
    \frac{1}{2}\rho_i|\bfu_i|^2
  \right)
  =
    \bfu_i \cdot \frac{\partial}{\partial t}(\rho_i\bfu_i)
  =
    \sum_j\rho_i\lambda_{i,j}\bfu_i\cdot(\bfu_{i,j}-\bfu_i)
.
    \end{equation}
With this, it is clear that the terms in \eqref{eq:energyRelaxation} involving the velocity cancel.
Thus, \eqref{eq:energyRelaxation} becomes
    \begin{equation}
    \label{eq:temperatureRelaxation}
\frac{\partial}{\partial t}
  \left(
    \frac{d}{2}n_iT_i
  \right)
  = 
    d\sum_jn_i\lambda_{i,j}(T_{i,j}-T_i)
,
    \end{equation}
a system of equations that describes the time evolution of the temperature moment system.
As above, we neglect species in $\{1,\cdots,N\}\setminus\{i,j\}$, take differences of relaxation rates in \eqref{eq:temperatureRelaxation}, and use \eqref{eq:T_ij_def} and \eqref{eq:gammaSymmetry} to obtain \eqref{eq:temperatureDiffRelax_MLB}.
    \end{proof}

    \begin{remark}
In this time-continuous case, the relaxation moment equation obtained for the energy relaxation is simpler than for the BGK model of \cite{Haack2017,asymptoticRelaxation}, since, as shown in \eqref{eq:d_tKineticEnergy}, the terms involving the bulk velocity cancel directly, whereas, these terms must remain as part of the energy equation in \cite{asymptoticRelaxation}.
    \end{remark}
    
    \bibliographystyle{plain}
    \bibliography{references}

    \end{document}